\documentclass[12pt]{article}
\usepackage{amsmath,amssymb,bbm,epsfig,wrapfig,xcolor}
\usepackage[left=1in, right=1in, top=1in, bottom=1in]{geometry}
\usepackage{lmodern}
\usepackage{slantsc}
\usepackage{ytableau}
\usepackage{sectsty}
\subsubsectionfont{\centering\fontsize{12}{15}\selectfont}
\subsectionfont{\centering\fontsize{12}{15}\selectfont}
\sectionfont{\centering\fontsize{14}{15}\selectfont}
\usepackage[english]{babel}
\usepackage[utf8]{inputenc}
\usepackage[mathscr]{eucal}
\usepackage[affil-it]{authblk}
\usepackage{verbatim}
\usepackage{enumitem}
\usepackage{xfrac}
\usepackage{amsthm}
\usepackage{amsfonts}
\usepackage{mathtools}
\usepackage[colorlinks]{hyperref}
\usepackage{graphicx}
\usepackage{subcaption}
\usepackage{caption}
\usepackage{appendix}
\usepackage{cleveref}
\usepackage[colorinlistoftodos]{todonotes}
\newcommand{\newc}{\newcommand}
\newc{\beq}{\begin{equation}}
\newc{\eeq}{\end{equation}}
\newc{\Tr}{\mbox{{\rm Tr}}}
\newc{\secu}{\mbox{{\rm Sc}}}

\newtheorem{theorem}{Theorem}

\newtheorem{lemma}{Lemma}[theorem]
\newtheorem{proposition}{Proposition}[theorem]
\numberwithin{theorem}{section}
\numberwithin{corollary}{section}
\numberwithin{lemma}{section}
\numberwithin{proposition}{section}
\numberwithin{figure}{section}			

\numberwithin{equation}{section}		
\usepackage{array}
\newcolumntype{b}{X}
\newcolumntype{s}{>{\hsize=.5\hsize}X}
\hypersetup{
    colorlinks=true,
    linkcolor=blue,
    filecolor=magenta,      
    citecolor=blue,
    urlcolor = blue
}

\catcode`,\active

\catcode`\,12

\begin{document}

\title{On the moments of characteristic polynomials}

\author[1]{ Bhargavi Jonnadula}
\affil[1]{\small School of Mathematics, University of Bristol, Fry Building,
Bristol,
BS8 1UG, UK}
\author[2]{Jonathan P. Keating}
\affil[2]{\small Mathematical Institute, University of Oxford, Andrew Wiles Building, Oxford, OX2 6GG, UK}
\author[1]{Francesco Mezzadri}

\date{}


\maketitle

\begin{abstract}
We examine the asymptotics of the moments of characteristic polynomials of $N\times N$ matrices drawn from the Hermitian ensembles of Random Matrix Theory, in the limit as $N\to\infty$. We focus in particular on the Gaussian Unitary Ensemble, but discuss other Hermitian ensembles as well. We employ a novel approach to calculate asymptotic formulae for the moments, enabling us to uncover subtle structure not apparent in previous approaches.  
\end{abstract}
\section{Introduction}
The characteristic polynomials of random matrices have received considerable attention over the past twenty years. One of the principal motivations stems from their connections to  the statistical properties of the Riemann zeta-function and other families of $L$-functions \cite{Keating2000random, Keating2000, Hughes2000, Conrey2001,Conrey2003,Conrey2005integral,Keating2011,Fyodorov2012,Fyodorov2014}.  In this context, the value distributions of the characteristic polynomials of random unitary, orthogonal and symplectic matrices have been calculated using a variety of approaches. For example, the moments have been computed in all three cases and the results used to develop conjectures for the moments of the Riemann zeta-function $\zeta(s)$ on its critical line and for the moments of families of $L$-functions at the centre of the critical strip.  Specifically, if $A$ is an $N\times N$ unitary matrix, drawn at random uniformly with respect to Haar measure on the unitary group $U(N)$, then for ${\rm Re}\beta>-1/2$
\begin{equation}
\mathbb{E}_{A\in U(N)}\big[|\det(I-Ae^{-i\theta})|^{2\beta}\big] =\prod_{j=1}^N \frac{\Gamma(j)\Gamma(j+2\beta)}
{\Gamma(j+\beta)^2}
\end{equation}
from which one can deduce that as $N\to\infty$
\begin{equation}
\mathbb{E}_{A\in U(N)}\big[|\det(I-Ae^{-i\theta})|^{2\beta}\big]\sim\frac{G(1+\beta)^2}{G(1+2\beta)}N^{\beta^2},
\end{equation}
where $G(s)$ is the Barnes $G$-function, and for  $k\in\mathbb{N}$
\begin{equation}
\label{KSCUE}
\mathbb{E}_{A\in U(N)}\big[|\det(I-Ae^{-i\theta})|^{2k}\big]\sim\left(\prod_{m=0}^{k-1}\frac{m!}{(m+k)!}\right)N^{k^2}.
\end{equation}
These formulae lead to the conjectures \cite{Keating2000random} that for ${\rm Re}\beta>-1/2$, as $T\rightarrow\infty$
\begin{equation}
\frac{1}{T}\int_0^T|\zeta\left(1/2+i t\right)|^{2\beta}{\rm d}t\sim a(\beta)\frac{G(1+\beta)^2}{G(1+2\beta)}\left(\log T\right)^{\beta^2}
\end{equation}
and for $k\in\mathbb{N}$, as $T\rightarrow\infty$
\begin{equation}
\frac{1}{T}\int_0^T|\zeta\left(1/2+i t\right)|^{2k}{\rm d}t\sim a(k)\prod_{m=0}^{k-1}\frac{m!}{(m+k)!}\left(\log T\right)^{k^2}
\end{equation}
where 
\begin{equation}
a(s)=\prod_p\left[\left(1-\frac{1}{p^s}\right)^{s^2}\sum_{m=0}^\infty\left(\frac{\Gamma(m+s)}{m!\Gamma(s)}\right)^2p^{-m}\right]
\end{equation}
with the product running over primes $p$.

Our focus here will primarily be on the Gaussian Unitary Ensemble (GUE) of random complex Hermitian matrices. For an $N\times N$ matrix $M$ drawn from the GUE, the joint eigenvalue density function is
\beq
\frac{1}{\mathscr{Z}_N^{(H)}}\prod_{1\leq j<k \leq N}|x_j-x_k|^2\prod_{j=1}^Ne^{-\frac{Nx_j^2}{2}},
\eeq
where ${\mathscr{Z}_N^{(H)}}$ is a normalization constant.
Brezin and Hikami \cite{Brezin2000}  calculated the $N\to\infty$ asymptotics of the moments of the associated characteristic polynomials to be
\beq\label{eq:BH0}
\mathbb{E}^{(H)}_N\left[\det(t-M)^{2p}\right] = e^{-Np}e^{Np\frac{t^2}{2}}(2\pi N\rho_{sc}(t))^{p^2}\prod_{j=0}^{p-1}\frac{j!}{(p+j)!},\quad p\in\mathbb{N},
\eeq
where the asymptotic eigenvalue density is given by the Wigner semi-circle law
\beq
\rho_{sc}(x) = \frac{1}{2\pi}\sqrt{4-x^2}.
\eeq
This corresponds precisely to \eqref{KSCUE}, where the mean density is constant.  

Our purpose is to focus on some subtle features of the asymptotics of the moments of the characteristic polynomials of GUE matrices, and of matrices drawn from other unitarily invariant ensembles, that are not captured by \eqref{eq:BH0}.  In particular, we will show that the asymptotics when $N$ is even differs from when it is odd, and that one only recovers \eqref{eq:BH0} when one formally averages the two cases.  At the one hand, this is somewhat surprising, because \eqref{eq:BH0} is a central pillar of the theory of the characteristic polynomials of random matrices, but on the other it is not completely expected, based on the following reasoning.  A straightforward application of the method of orthogonal polynomials gives that the correlations of characteristic polynomials have a determinantal structure involving classical orthogonal polynomials. For the GUE, for example,
\beq
\mathbb{E}^{(H)}_N\big[\prod_{j=1}^{2p}\det(t_j-M)\big] = \frac{1}{\Delta(\textbf{t})}\det[N^{-\frac{N-j}{2}-p}H_{N+2p-j}(\sqrt{N}t_k)]_{1\leq j,k \leq 2p},
\eeq
where $H_n(x)$ is a Hermite polynomial of degree $n$ and 
\beq
\Delta(\textbf{t})=\Delta(t_1,\dots ,t_{2p}) = \prod_{1\leq j<k\leq 2p}(t_j-t_k)
\eeq
is the Vandermonde determinant. As a consequence, the moments also take a determinantal form comprising derivatives of Hermite polynomials. But Hermite polynomials depend on the parity of the degree  $n$ via
\beq
H_n(-x) = (-1)^n H_n(x).
\eeq
Therefore, the moments of characteristic polynomials also depend on the parity of the degree, which in turn depends on the parity of the dimension of the matrix $N$. As a result, one might expect that the asymptotic behaviour of the moments of characteristic polynomials should be different for even and odd dimensional matrices. The question is: at what order in the asymptotics is this important?  We show below that for the GUE it influences the leading-order behaviour. This is in contrast to other Hermitian ensembles such as the Laguerre unitary ensemble (LUE) and the Jacobi unitary ensemble (JUE), for which both even and odd dimensional matrices have the same moments at leading order.

Brezin and Hikami \cite{Brezin2000} used orthogonal polynomial techniques to arrive at \eqref{eq:BH0}. Other studies to-date relating to the asymptotics of the moments of characteristic polynomials have relied mainly on the orthogonal polynomial method and saddle point techniques \cite{Brezin2000,Brezin2000edge,Baik2003}, the Riemann-Hilbert method \cite{Strahov2003}, Hankel determinants with Fisher-Hartwig symbols \cite{Krasovsky2007,Forrester2004,Garoni2005}, and  supersymmetric tools \cite{Andreev1995,Fyodorov2002,Fyodorov2002negative,Szabo2001}. In the present paper we take a different line of attack: we express the moments in terms of certain multivariate orthogonal polynomials and take a combinatorial approach to compute the asymptotics of the moments using the properties of these polynomials. By doing so, we discover that even and odd dimensional GUE matrices give different contributions in the large $N$ limit, and that only a formal average gives formulae consistent with \eqref{eq:BH0}. In Sec.~\ref{sec:second moment}, this phenomenon is discussed in detail for the second moment of the characteristic polynomial.

In addition to connections with number theory, characteristic polynomials have found numerous applications in quantum chaos \cite{Andreev1995}, mesoscopic systems \cite{Fyodorov1995}, quantum chromodynamics \cite{Damgaard1998}, and in a variety of combinatorial problems \cite{Strahov2003moments,Diaconis2004}. The asymptotic study of negative moments and ratios of characteristic polynomials is another active area of research, see for example \cite{Berry2001,Fyodorov2003negative,Fyodorov2003,Baik2003,Strahov2003,Forrester2004singularity,Borodin2006,Breuer2012,Fyodorov2018,Akemann2020}. More recently, the statistics of the maximum of the characteristic polynomial are being extensively studied, motivated by the relations to logarithmically correlated Gaussian processes. For example, see \cite{Fyodorov2012,Fyodorov2014,Fyodorov2016moments,Fyodorov2016} and references therein.  We expect that the techniques developed here will have applications to those calculations as well. 

This paper is structured as follows. After introducing the required tools in Sec.~\ref{sec:background}, we recall the moments of characteristic polynomials of the GUE, LUE and JUE in Sec.~\ref{sec:moments}. In Sec.~\ref{sec:asymptotics}, we compute the asymptotics of moments of the GUE and illustrate how to recover the semi-circle law in the limit as the matrix size goes to infinity. In the last section Sec.~\ref{sec:secular coeff}, as an application of the results discussed, we compute the correlations of secular coefficients which are the coefficients of a characteristic polynomial when expanded as a function of the spectral variable. 

\section{Background}\label{sec:background}
A partition $\mu$ is a sequence of integers $(\mu_1,\dots,\mu_l)$ such that $\mu_1\geq\dots\geq\mu_l>0$. Here $l$ is the length of the partition and we denote $|\mu|=\mu_1+\dots+\mu_l$ to be the weight of the partition. We do not distinguish partitions that only differ by a sequence of zeros at the end. For example $(4,2)$ and $(4,2,0,0)$ are equivalent with length $l=2$ and weight 6. A partition can be represented with a \textit{Young diagram} which is a left adjusted table of $|\mu|$ boxes and $l(\mu)$ rows such that the first row contains $\mu_1$ boxes, the second row contains $\mu_2$ boxes, and so on. The conjugate partition $\mu^\prime$ is defined by transposing the Young diagram of $\mu$ along the main diagonal. 
\beq
\begin{split}
\ytableausetup{boxsize=1.2em}
\ydiagram[]{4,2} &\quad\qquad\qquad \ydiagram[]{2,2,1,1}\\
\text{Young diagram of $\mu$}&\qquad\text{Young diagram of $\mu^\prime$}
\end{split}
\eeq

For a partition $\mu$, let $\varPhi_\mu$ be the multivariate symmetric polynomial, with leading coefficient equal to 1, that obey the orthogonality relation
\beq\label{eq:mul_poly}
\int \varPhi_\mu(x_1,\dots,x_N)\varPhi_\nu(x_1,\dots,x_N)\prod_{1\leq i<j\leq N}(x_i-x_j)^2\prod_{j=1}^Nw(x_j)\,dx_j = \delta_{\mu\nu}C_\mu
\eeq
for a weight function $w(x)$. Here the lengths of the partitions $\mu$ and $\nu$ are less than or equal to the number of variables $N$, and $C_\mu$ is a constant which depends on $N$. Polynomial $\varPhi_\mu$ can be expressed as a ratio of determinants, as given in \cite{Sergeev2014},
\beq\label{eq:gen_poly}
\begin{split}
\varPhi_\mu(\textbf{x}) = \frac{1}{\Delta(\textbf{x})}
\begin{vmatrix}
\varphi_{\mu_1+N-1}(x_1) & \varphi_{\mu_1+N-1}(x_2) & \dots & \varphi_{\mu_1+N-1}(x_N)\\
\varphi_{\mu_2+N-2}(x_1) & \varphi_{\mu_2+N-2}(x_2) & \dots & \varphi_{\mu_2+N-2}(x_N)\\
\vdots & \vdots & & \vdots\\
\varphi_{\mu_N}(x_1) & \varphi_{\mu_N}(x_2) & \dots & \varphi_{\mu_N}(x_N)
\end{vmatrix},
\end{split}
\eeq
where $\varphi_j(x)$ is a polynomial of degree $j$ orthogonal with respect to $w(x)$. We focus in particular to the case when $w(x)$ in \eqref{eq:mul_poly} is one of the weights
\beq\label{eq:weights}
w(x) = 
\begin{cases}
e^{-\frac{N x^2}{2}}, \qquad\qquad\,\,\,\, x\in \mathbb{R}, \qquad\qquad\qquad\quad\quad\,\,\, \text{Gaussian},\\
x^\gamma e^{-2Nx},  \qquad\quad\,\,\,\, x\in \mathbb{R}_+,\qquad \gamma>-1, \qquad\,\,\text{Laguerre},\\
x^{\gamma_1}(1-x)^{\gamma_2}, \qquad x\in[0,1], \quad \gamma_1,\gamma_2>-1, \quad \text{Jacobi}.
\end{cases}
\eeq
The monic polynomials orthogonal with respect to these weights are
\begin{align}
    h_n(x) &= N^{-\frac{n}{2}}H_n(\sqrt{N}x),\\
    l_n^{(\gamma)}(x)&=\frac{(-1)^nn!}{(2N)^n}L_n^{(\gamma)}(2Nx),\\
    j_n^{(\gamma_1,\gamma_2)}&=(-1)^nn!\frac{\Gamma(n+\gamma_1+\gamma_2+1)}{\Gamma(2n+\gamma_1+\gamma_2+1)}J_n^{(\gamma_1,\gamma_2)}(x),
    \end{align}
    where $H_n(x)$, $L_n^{(\gamma)}(x)$ and $J_n^{(\gamma_1,\gamma_2)}(x)$ are the classical orthogonal polynomials that satisfy
\begin{subequations}
\label{eq:classical_hlj_ortho}
\begin{gather}
\label{eq:classical_hlj_ortho1}
\int_{\mathbb{R}}H_j(x)H_k(x)e^{-\frac{x^2}{2}}\,dx = \sqrt{2\pi}j!\delta_{jk},\\
\label{eq:classical_hlj_ortho2}
 \quad\int_{\mathbb{R}_{+}}L^{(\gamma)}_mL^{(\gamma)}_nx^\gamma e^{-x}\, dx 
= \frac{\Gamma(n+\gamma +1)}{\Gamma(n+1)}\delta_{nm},\\
\int_0^1 J^{(\gamma_1,\gamma_2)}_n(x)J^{(\gamma_1,\gamma_2)}_m(x)x^{\gamma_1}(1-x)^{\gamma_2}\,dx \nonumber\\
\label{eq:classical_hlj_ortho3}
= \frac{1}{(2n+\gamma_1+\gamma_2+1)}\frac{\Gamma(n+\gamma_1+1)\Gamma(n+\gamma_2+1)}{n!\Gamma(n+\gamma_1+\gamma_2+1)}\delta_{mn}.
\end{gather}
\end{subequations}

When $\varphi_n(x)$ in \eqref{eq:gen_poly} is chosen to be one of the  Hermite $h_n(x)$, Laguerre $l_n^{(\gamma)}(x)$ and Jacobi $j_n^{(\gamma_1,\gamma_2)}(x)$ polynomials of degree $n$, we get their multivariable analogues denoted by $\mathscr{H}_\mu$, $\mathscr{L}^{(\gamma)}_\mu$ and $\mathscr{J}^{(\gamma_1,\gamma_2)}_\mu$. These multivariate generalisations are the eigenfunctions of differential equations called Calogero–Sutherland Hamiltonians. Several properties such as recursive relations and integration formulas extend to the multivariate case \cite{Baker1997,Baker1997calogero}.

Define
\beq\label{eq:C_lambda_and_G_lambda}
\begin{split}
C_\lambda(N) &= \prod_{j=1}^N\frac{(\lambda_j+N-j)!}{(N-j)!},\\
G_\lambda(N,\gamma) &= \prod_{j=1}^N\Gamma(\lambda_j+N-j+\gamma +1).
\end{split}
\eeq
The constants $C_\lambda(N)$ and $G_\lambda(N,\gamma)$ have several interesting combinatorial interpretations \cite{Keating2010}. The joint probability densities function for the GUE, LUE and JUE are
\begin{align}
    p^{(H)}(x_1,\dots,x_N)&=\frac{1}{\mathscr{Z}_N^{(H)}}\Delta^2(x_1,\dots,x_N)\prod_{j=1}^Ne^{-\frac{Nx_j^2}{2}},\\
    p^{(L)}(x_1,\dots,x_N)&=\frac{1}{\mathscr{Z}_N^{(L)}}\Delta^2(x_1,\dots,x_N)\prod_{j=1}^Nx_j^\gamma e^{-2Nx_j},\\
    p^{(J)}(x_1,\dots,x_N)&=\frac{1}{\mathscr{Z}_N^{(J)}}\Delta^2(x_1,\dots,x_N)\prod_{j=1}^Nx_j^{\gamma_1}(1-x_j)^{\gamma_2},
\end{align}
with
\begin{align}
   \mathscr{Z}_N^{(H)}&= \frac{(2\pi)^{\frac{N}{2}}}{N^{\frac{N^2}{2}}}\prod_{j=1}^Nj!, \\
   \mathscr{Z}_N^{(L)}&= \frac{N!}{(2N)^{N(N+\gamma)}}G_0(N,\gamma)G_0(N,0),\\
  \mathscr{Z}_N^{(J)}&= N!\prod_{j=0}^{N-1}\frac{j!\,\Gamma(j+\gamma_1+1)\Gamma(j+\gamma_2+1)\Gamma(j+\gamma_1+\gamma_2+1)}{\Gamma(2j+\gamma_1+\gamma_2+2)\Gamma(2j+\gamma_1+\gamma_2+1)}.
\end{align}

Similar  to \eqref{eq:classical_hlj_ortho}, the polynomials $\mathscr{H}_\mu$, $\mathscr{L}_\mu^{(\gamma)}$ and $\mathscr{J}_\mu^{(\gamma_1,\gamma_2)}$ satisfy
\begin{align}
&\frac{1}{\mathscr{Z}_n^{(H)}}\int_{(-\infty,\infty)^n}\mathscr{H}_\mu(\textbf{x})\mathscr{H}_\nu(\textbf{x})\Delta^2(\textbf{x})\prod_{j=1}^n e^{-\frac{Nx_j^2}{2}}\, dx_j = \frac{1}{N^{|\mu|}}C_\mu(n)\delta_{\mu\nu},\\
&\frac{1}{\mathscr{Z}_n^{(L)}}\int_{[0,\infty)^n}\mathscr{L}_\mu(\textbf{x})\mathscr{L}_\nu(\textbf{x})\Delta^2{(\textbf{x})}\prod_{j=1}^nx_j^\gamma e^{-2Nx_j}\,dx_j = \frac{1}{(2N)^{2|\mu|}}\frac{G_\mu(n,\gamma)}{G_0(n,\gamma)}C_\mu(n)\delta_{\mu\nu},\label{eq:ortho mullag}\\
&\frac{1}{\mathscr{Z}^{(J)}_n}\int_{[0,1]^n}\mathscr{J}^{(\gamma_1,\gamma_2)}_\mu(\textbf{x}) \mathscr{J}^{(\gamma_1,\gamma_2)}_\nu(\textbf{x})\Delta^2(\textbf{x})\prod_{j=1}^nx_j^{\gamma_1}(1-x_j)^{\gamma_2}\,dx_j\nonumber \\
=&\prod_{j=1}^N\frac{\Gamma(2n-2j+\gamma_1+\gamma_2+1)\Gamma(2n-2j+\gamma_1+\gamma_2+2)}{\Gamma(2\lambda_j+2n-2j+\gamma_1+\gamma_2+1)\Gamma(2\lambda_j+2n-2j+\gamma_1+\gamma_2+2)}\nonumber\\
&\quad\times \frac{G_\mu(n,\gamma_1+\gamma_2)G_\mu(n,\gamma_1)G_\mu(n,\gamma_2)}{G_0(n,\gamma_1+\gamma_2)G_0(n,\gamma_1)G_0(n,\gamma_2)}C_\mu(n)\delta_{\lambda\mu},
\end{align}

The Schur polynomials $S_\lambda$, indexed by a partition $\lambda$, are defined as 
\beq
S_\lambda(x_1,\dots,x_n) = \frac{\det[x_k^{\lambda_j+n-j}]_{1\leq j,k\leq n}}{\det[x_k^{n-j}]_{1\leq j,k\leq n}},
\eeq
for $l(\lambda)\leq n$, and $S_\lambda=0$ for $l(\lambda)>n$. The polynomials $\mathscr{H}_\mu$, $\mathscr{L}^{(\gamma)}_\mu$ and $\mathscr{J}^{(\gamma_1,\gamma_2)}_\mu$ form a basis for symmetric polynomials of degree $|\mu|$.
The Schur polynomials can be expanded as \cite{Jonnadula2020}
\beq\label{eq:schur to mulher}
\begin{split}
S_\lambda(x_1,\dots,x_n) = \sum_{\nu\subseteq\lambda}\Psi_{\lambda\nu}\varPhi_\nu(x_1,\dots,x_n),
\end{split}
\eeq
where $\varPhi_\nu(\bf x)$ can be either $\mathscr{H}_\nu$, $\mathscr{L}^{(\gamma)}_\nu$, or $\mathscr{J}^{(\gamma_1,\gamma_2)}_\nu$. In the following the superscripts $(H)$, $(L)$ and $(J)$ indicate Hermite, Laguerre and Jacobi, respectively. The coefficients in \eqref{eq:schur to mulher} are
\beq
\Psi_{\lambda\nu}^{(H)}=\left(\frac{1}{2N}\right)^{\frac{|\lambda|-|\nu|}{2}}\frac{C_\lambda(n)}{C_\nu(n)}D_{\lambda\nu}^{(H)},
\eeq
where 
\beq\label{eq:D her}
D_{\lambda\nu}^{(H)} = \det\left[\mathbbm{1}_{\lambda_j-\nu_k-j+k= \text{0 mod 2}}\,\frac{1}{\left(\frac{\lambda_j-\nu_k-j+k}{2}\right)!}\right]_{j,k=1,\dots , l(\lambda)}.
\eeq
Similarly, one has
\begin{align}
\Psi_{\lambda\nu}^{(L)} &= \frac{1}{(2N)^{|\lambda|-|\nu|}}\frac{G_\lambda(n,\gamma)G_\lambda(n,0)}{G_\nu(n,\gamma)G_\nu(n,0)}D_{\lambda\mu}^{(L)},\\
\Psi_{\lambda\nu}^{(J)} &= \frac{G_\lambda(n,\gamma_1)G_\lambda(n,0)}{G_\nu(n,\gamma_1)G_\nu(n,0)}\left(\prod_{j=1}^n\Gamma(2\nu_j+2n-2j+\gamma_1+\gamma_2+2)\right)\mathcal{D}^{(J)}_{\lambda\nu},\label{eq: schur to mul jac coef}
\end{align}
where 
\begin{align}
&D^{(L)}_{\lambda\nu}= \text{det}\left[\mathbbm{1}_{\lambda_i-\nu_j-i+j\geq 0\frac{1}{(\lambda_i-\nu_j-i+j)!}}\right]_{1\leq i,j\leq l(\lambda)}\label{eq:D lag},\\
&\mathcal{D}^{(J)}_{\lambda\nu}=\text{det}\left[\mathbbm{1}_{\lambda_j-\nu_k-j+k\geq 0}\frac{1}{(\lambda_j-\nu_k-j+k)!\,\Gamma(2n+\lambda_j+\nu_k-j-k+\gamma_1+\gamma_2+2)}\right]_{1\leq i,j\leq n}\label{eq:D jac}.
\end{align}

The polynomials $\mathscr{H}_\mu$, $\mathscr{L}^{(\gamma)}_\mu$ and $\mathscr{J}^{(\gamma_1,\gamma_2)}_\mu$ are chosen such that the leading coefficient of these polynomials in the Schur basis is 1. More precisely,
\beq
\varPhi_\lambda(x_1,\dots,x_n) = \sum_{\mu\subseteq\lambda}\Upsilon_{\lambda\mu}S_\mu(x_1,\dots,x_n),
\eeq
where $\varPhi_\lambda$ is one of the $\mathscr{H}_\lambda$, $\mathscr{L}^{(\gamma)}_\lambda$, $\mathscr{J}^{(\gamma_1,\gamma_2)}_\lambda$, and 
\begin{align}
\Upsilon_{\lambda\mu}^{(H)}&=\left(\frac{-1}{2N}\right)^{\frac{|\lambda|-|\mu|}{2}}\frac{C_\lambda(n)}{C_\mu(n)}D_{\lambda\mu}^{(H)},\\
\Upsilon_{\lambda\mu}^{(L)}&=\left(\frac{-1}{2N}\right)^{|\lambda|-|\mu|}\frac{G_\lambda(n,\gamma)G_\lambda(n,0)}{G_\mu(n,\gamma)G_\mu(n,0)}D_{\lambda\mu}^{(L)},\\
\Upsilon_{\lambda\mu}^{(J)}&=(-1)^{|\lambda|+|\mu|}\left(\prod_{j=1}^n\frac{1}{\Gamma(2\lambda_j+2n-2j+\gamma_1+\gamma_2+1)}\right)\frac{G_\lambda(n,\gamma_1)G_\lambda(n,0)}{G_\mu(n,\gamma_1)G_\mu(n,0)}\tilde{\mathcal{D}}^{(J)}_{\lambda\mu},
\end{align}
with
\beq
\tilde{\mathcal{D}}^{(J)}_{\lambda\nu}=\text{det}\left[\mathbbm{1}_{\lambda_j-\nu_k-j+k\geq 0}\frac{\Gamma(2n+\lambda_j+\nu_k-j-k+\gamma_1+\gamma_2+1)}{(\lambda_j-\nu_k-j+k)!}\right]_{1\leq i,j\leq n}
\eeq
In this paper, the above results play an important role in studying the correlations of characteristic polynomials and  secular coefficients.

\section{Moments of characteristic polynomials}\label{sec:moments}
The Schur polynomials satisfy the following identity which has proven to be crucial in computing the correlations of characteristic polynomials of the unitary group \cite{Bump2006}.
\begin{lemma}[Dual Cauchy identity]
Let $p,N\in\mathbb{N}$. For $\lambda\subseteq (N^p)\equiv (\underbrace{N,\dots,N}_{p})$, let $\tilde{\lambda}=(p-\lambda_N^\prime,\dots,p-\lambda_1^\prime)$. Then \cite{Macdonald1998}
\beq
\prod_{i=1}^p\prod_{j=1}^N(t_i-x_j) = \sum_{\lambda\subseteq (N^p)} (-1)^{|\tilde{\lambda}|}S_\lambda(t_1,\dots,t_p)S_{\tilde{\lambda}}(x_1,\dots ,x_N).
\eeq
\end{lemma}
\noindent Here $\lambda=(\lambda_1,\dots,\lambda_p)$ is a sub-partition of $(N^p)$ indicated by $\lambda\subseteq (N^p)$ (each $\lambda_j\leq N$ for $j=1,\dots, p$) and $\lambda^\prime$ is the conjugate partition of $\lambda$. The $\varPhi_\mu$'s satisfy a  generalised dual Cauchy identity, which is similar to that for the Schur polynomials.
\begin{lemma}\label{lemma:dualcauchy_to_mulmonicpoly}
With the notation introduced above, we have \cite{Jonnadula2020}.
\beq\label{eq:new_identity}
\prod_{i=1}^p\prod_{j=1}^N(t_i-x_j) = \sum_{\lambda\subseteq (N^p)} (-1)^{|\tilde{\lambda}|}\varPhi_\lambda(t_1,\dots,t_p)\varPhi_{\tilde{\lambda}}(x_1,\dots ,x_N).
\eeq
\end{lemma}
\noindent The identity in \eqref{eq:new_identity} gives a compact way to calculate the correlation functions and moments of characteristic polynomials of unitary invariant Hermitian ensembles.

\begin{proposition}\label{prop:correlations_charpoly}
Let $M$ be an $N\times N$ GUE, LUE or JUE matrix and $t_1,\dots, t_p\in\mathbb{C}$. Then, using the generalised dual Cauchy identity \cite{Jonnadula2020},
\begin{align}
\textit{(a)}\quad\mathbb{E}^{(H)}_N\big[\prod_{j=1}^p\det(t_j - M)\big] &=\mathscr{H}_{(N^p)}(t_1,\dots, t_p)\label{eq:char poly corre her rescaled}\\
\textit{(b)}\quad\mathbb{E}^{(L)}_N\big[\prod_{j=1}^p\det(t_j-M)\big] &= \mathscr{L}^{(\gamma)}_{(N^p)}(t_1,\dots,t_p)\\
\textit{(c)}\quad\mathbb{E}^{(J)}_N\big[\prod_{j=1}^p\det(t_j-M)\big] &= \mathscr{J}^{(\gamma_1,\gamma_2)}_{(N^p)}(t_1,\dots,t_p)
\end{align}
\end{proposition}
The moments can be readily computed from the above formulae by taking the limit $t_j\rightarrow t$ for $j=1,\dots ,p$. This leads to a determinantal formula for the moments involving the derivatives of orthogonal polynomials:
\beq\label{eq:moments derivatives}
\frac{(-1)^{\frac{1}{2}p(p-1)}}{\prod_{j=1}^{p-1}(j-1)!}
\begin{vmatrix}
\varphi_N(t) &\varphi_{N+1}(t) &\dots &\varphi_{N+p-1}(t)\\
\varphi^\prime_N(t) &\varphi^\prime_{N+1}(t) &\dots &\varphi^\prime_{N+p-1}(t)\\
\vdots &\vdots & &\vdots\\
\varphi^{(p-1)}_N(t) &\varphi^{(p-1)}_{N+1}(t) &\dots &\varphi^{(p-1)}_{N+p-1}(t)
\end{vmatrix}.
\eeq
Here $\varphi_n(t)$ are Hermite $h_n(t)$, Laguerre $l_n^{(\gamma)}(t)$ and Jacobi $j_n^{(\gamma_1,\gamma_2)}(t)$ polynomials for the GUE, LUE and JUE, respectively. By expressing the multivariate polynomials in the Schur basis and using 
\beq\label{eq:c lam n}
\begin{split}
C_\lambda(n) &=  \prod_{j=1}^{l(\lambda)}\frac{(\lambda_j+n-j)!}{(n-j)!} = |\lambda|!\frac{S_\lambda(1^n)}{\text{dim $V_\lambda$}},
\end{split}
\eeq
where the dimension of the irreducible representation of the symmetric group is
\beq
\text{dim}\,V_\lambda = |\lambda|!\frac{\prod_{1\leq j<k\leq l(\lambda)}\lambda_j-\lambda_k-j+k}{\prod_{j=1}^{l(\lambda)}(\lambda_j+l(\lambda)-j)!},
\eeq
we have the following proposition.


\begin{proposition}\label{prop:momemts char poly}
Let $\lambda=(N^{p})$. The moments of characteristic polynomial are given by \cite{Jonnadula2020}
\begin{align}
\mathbb{E}^{(H)}_N\left[\det(t-M)^p\right] &= C_\lambda(p)\sum_{\nu\subseteq \lambda}\left(\frac{-1}{2N}\right)^{\frac{|\lambda|-|\nu|}{2}}\frac{\dim V_\nu}{|\nu|!} D^{(H)}_{\lambda\nu}t^{|\nu|}\label{eq:mom_char_poly_gue}\\
\mathbb{E}^{(L)}_N\left[\det(t-M)^p\right] &=\left(\frac{-1}{2N}\right)^{Np}\frac{G_\lambda(p,\gamma)G_\lambda(p,0)}{G_0(p,0)}\sum_{\nu\subseteq \lambda}\frac{(-2N)^{|\nu|}}{G_\nu(p,\gamma)}\frac{\text{dim}\,V_\nu}{|\nu|!}D_{\lambda\nu}^{(L)}t^{|\nu|}\label{eq:mom_char_poly_lue}\\
\mathbb{E}^{(J)}_N\left[\det(t-M)^p\right] &= \left(\prod_{j=N}^{N+p-1}\frac{1}{\Gamma(2j+\gamma_1+\gamma_2+1)}\right)(-1)^{Np}\frac{G_\lambda(p,\gamma_1)G_\lambda(p,0)}{G_0(p,0)}\nonumber\\
&\quad\times \sum_{\nu\subseteq\lambda}\frac{(-1)^{|\nu|}}{|\nu|!\,G_\nu(p,\gamma_1)}\dim V_\nu \tilde{\mathcal{D}}^{(J)}_{\lambda\nu}t^{|\nu|}\label{eq:mom_char_poly_jue}
\end{align}
\end{proposition}

These expansions are exact and hold for any finite $N$. In the next section we study the large $N$ asymptotics of the moments of characteristic polynomials. In the large $N$ limit, only the even moments are interesting, since the odd moments result in oscillatory behaviour.


\section{Asymptotics}\label{sec:asymptotics}
In this section, we consider the asymptotics of the moments of characteristic polynomials for the GUE. By exploiting the integral representation of the classical Hermite polynomials, Brezin and Hikami \cite{Brezin2000} showed that in the Dyson limit, $N\rightarrow\infty$, $t_i-t_j\rightarrow 0$ and $N(t_i-t_j)$ finite, the moments of characteristic polynomials are
\beq\label{eq:BH}
\mathbb{E}^{(H)}_N\left[\det(t-M)^{2p}\right] = e^{-Np}e^{Np\frac{t^2}{2}}(2\pi N\rho_{sc}(t))^{p^2}\prod_{j=0}^{p-1}\frac{j!}{(p+j)!},
\eeq
where the asymptotic eigenvalue density is 
\beq
\rho_{sc}(x) = \frac{1}{2\pi}\sqrt{4-x^2}.
\eeq
Using \eqref{eq:mom_char_poly_gue}, we show in Sec.~\ref{sec:at center} that
\beq
\mathbb{E}^{(H)}_N\left[\det M^{2p}\right] =  e^{-Np}(2N)^{p^2}\prod_{j=0}^{p-1}\frac{j!}{(p+j)!},
\eeq
which coincides with \eqref{eq:BH} for $t=0$. 
For $t\neq 0$, we discover that the asymptotic behaviour is different for even and odd dimensional GUE matrices and these contributions combine in a special way to produce the semi-circle law. These cases are discussed in Sec.~\ref{sec:at center} and Sec.~\ref{sec:outside center} in more detail.


\subsection{Centre of the semi-circle}\label{sec:at center}
Let $\lambda=(N^{2p})$. For any finite $N$ we have
\beq\label{eq:char_mom_t0}
\mathbb{E}_N^{(H)}\left[\det M^{2p}\right] = \left(-\frac{1}{2N}\right)^{Np}C_{\lambda}(2p)D^{(H)}_{\lambda 0}.
\eeq
\begin{proposition}
\beq\label{eq:D_lambda0}
\begin{split}
D^{(H)}_{\lambda 0} &= \prod_{j=0}^{p-1} \frac{j!^2}{(m+j)!^2}, \quad N=2m,\,\, m\in\mathbb{N},\\
D^{(H)}_{\lambda 0} &= (-1)^p\frac{m!}{(m+p)!}\prod_{j=0}^{p-1} \frac{j!^2}{(m+j)!^2}, \quad N=2m+1,\quad m\in\mathbb{N}.
\end{split}
\eeq
\end{proposition}
\begin{proof}
The determinant $D^{(H)}_{\lambda 0}$ can be evaluated as follows. Let $N=2m$, then 
\beq
\begin{split}
D^{(H)}_{\lambda 0}&=\det\left[\mathbbm{1}_{k-j=\text{0 mod 2}}\left((m+\tfrac{k-j}{2})!\right)^{-1}\right]_{1\leq j,k\leq p}\\
&= \prod_{j=0}^{p-1}\frac{1}{(m+j)!^2}
\begin{vmatrix}
1 & 0 & m & 0 &\dots &\frac{m!}{(m-p+1)!} & 0\\
0 & 1 & 0 & m &\dots & 0 & \frac{m!}{(m-p+1)!}\\
1 & 0 & m+1 & 0 &\dots & \frac{(m+1)!}{(m-p+2)!} &0\\
0 & 1 & 0 & m+1 &\dots &0 &\frac{(m+1)!}{(m-p+2)!}\\
  &   &   &     &\vdots &  &  \\
1 & 0 &m+p-1 & 0 &\dots &\frac{(m+p-1)!}{m!} &0\\
0 & 1 & 0 &m+p-1 &\dots &0 &\frac{(m+p-1)!}{m!}
\end{vmatrix}.
\end{split}
\eeq
Perform the row operations $R_{2j}=R_{2j}-R_{2j-2}$, $R_{2j-1}=R_{2j-1}-R_{2j-3}$ with $j$ running from $p,p-1,\dots,2$ in that order. Using the Pascal's rule for  binomial coefficients, we get
\beq
D^{(H)}_{\lambda 0} = (p-1)!^2\prod_{j=0}^{p-1}\frac{1}{(m+j)!^2}
\begin{vmatrix}
1 & 0 & m & 0 &\dots &\frac{m!}{(m-p+1)!} & 0\\
0 & 1 & 0 & m &\dots & 0 & \frac{m!}{(m-p+1)!}\\
0 & 0 & 1 & 0 &\dots & \frac{m!}{(m-p+2)!} & 0\\
0 & 0 & 0 & 1 &\dots & 0 & \frac{m!}{(m-p+2)!} \\
  &   &   &     &\vdots &  &  \\
0 & 0 & 1 & 0 &\dots & \frac{(m+p-2)!}{m!} & 0\\
0 & 0 & 0 & 1 &\dots & 0 & \frac{(m+p-2)!}{m!}
\end{vmatrix}.
\eeq
Next perform $R_{2j}=R_{2j}-R_{2j-2}$, $R_{2j-1}=R_{2j-1}-R_{2j-3}$ with $j$ running from $p,p-1,\dots,3$ in that order. Repeat this process $p-2$ more times to reach an upper triangular matrix with determinant given in  \eqref{eq:D_lambda0}. Similarly, $D^{(H)}_{\lambda 0} $ can be calculated for $N$ odd. 
\end{proof}
Define
\beq\label{eq:De_Do}
\begin{split}
&D_e(N)= \prod_{j=0}^{p-1}\frac{j!^2}{(m+j)!^2},\quad N=2m,\\
&D_o(N) = (-1)^p\frac{m!}{(m+p)!}\prod_{j=0}^{p-1}\frac{j!^2}{(m+j)!^2},\quad N=2m+1.
\end{split}
\eeq
Using this notation, \eqref{eq:char_mom_t0} reads 
 \beq\label{eq:char_mom_t0_final}
 \mathbb{E}^{(H)}_N\left[\det M^{2p}\right] = \left(-\frac{1}{2N}\right)^{Np}\times
 \begin{cases}
 C_\lambda(2p) D_e(N),\quad \text{$N$ even},\\
 C_\lambda(2p)D_o(N),\quad \text{$N$ odd}.
 \end{cases}
 \eeq
The functions $C_\lambda(2p) D_e(N)$ and $C_\lambda(2p) D_o(N)$, $\lambda=(N^{2p})$, can be expressed in terms of the ratios of factorials,
\beq\label{eq:C_lambda_D}
\begin{split}
C_{(N^{2p})}(2p) D_e(N) &= \prod_{j=0}^{p-1}\frac{(2m+j)!(2m+p+j)!}{\left(m+j\right)!^2}\frac{j!}{(p+j)!},\quad N=2m,\\
C_{(N^{2p})}(2p) D_o(N) &= (-1)^p\frac{m!}{(m+p)!}\prod_{j=0}^{p-1}\frac{(2m+1+j)!(2m+1+p+j)!}{\left(m+j\right)!^2}\frac{j!}{(p+j)!},\quad N=2m+1.
\end{split}
\eeq 
Denote
\beq
\gamma_p = \prod_{j=0}^{p-1}\frac{j!}{(p+j)!}.
\eeq 
The universal constant $\gamma_p$ is present in the moments for any finite $N$. To compute the large $N$ limit, we require the asymptotic expansion of \eqref{eq:C_lambda_D}. In App.~\ref{app:asymp_ratio_factorials}, we compute the first few terms in this expansion. As $N\rightarrow\infty$,
\beq\label{eq:asymp_ratio_factorials}
\begin{split}
&C_{(N^{2p})}(2p)D_e(N) \sim e^{-Np}(2N)^{Np+p^2}\gamma_p\left[1+\frac{p}{6N}(4p^2+1) + O(N^{-2})\right], \quad N\,\, \text{even},\\
&C_{(N^{2p})}(2p)D_o(N)\sim (-1)^pe^{-Np}(2N)^{Np+p^2}\gamma_p\left[1 + \frac{p}{3N}(2p^2-1) + O(N^{-2})\right], \, N\,\, \text{odd}.
\end{split}
\eeq
Plugging \eqref{eq:asymp_ratio_factorials} in \eqref{eq:char_mom_t0_final},  the leading order  behaviour of the moments for $N$ even and $N$ odd is 
\beq
\begin{split}
&e^{-Np}(2N)^{p^2}\gamma_p
\end{split}
\eeq
which coincides with \eqref{eq:BH} for $t=0$. On the other hand, the sub-leading behaviour depends on the parity of $N$.

\subsection{Away from the centre of the semi-circle}\label{sec:outside center}
For $t_j=t$,
\beq\label{eq:char_mom}
\begin{split}
 \mathbb{E}^{(H)}_N\left[\det(t-M)^{2p}\right] &= C_\lambda(2p)\sum_{\nu\subseteq\lambda}\left(-\frac{1}{2N}\right)^{\frac{|\lambda|-|\nu|}{2}}\frac{\text{dim}\, V_\nu}{|\nu|!} D^{(H)}_{\lambda\nu}t^{|\nu|}.
 \end{split}
 \eeq
To compute the asymptotics near the centre of the semi-circle, $t\neq 0$, we need to evaluate $D^{(H)}_{\lambda\nu}$ for a non-empty partition $\nu$. In Table.~\ref{table:D_lambda_nu}, we give the values of $D^{(H)}_{\lambda\nu}$ when $\nu$ is a partition of 2 and 4.
\begin{table}[h!]
  \centering
    \begin{tabular}{ l  p{5cm}  p{5cm} }
    \hline\hline
     $D^{(H)}_{\lambda\nu}$ & N=2m &N=2m+1\\\hline
    $D^{(H)}_{\lambda 0}$
     & $D_e$ & $D_o$\\[2ex]
    $D^{(H)}_{\lambda(2)}$ & $mpD_e$ &$mpD_o$\\[2ex]
    $D^{(H)}_{\lambda(1^2)}$ & $-mpD_e$ &$-(m+1)pD_o$\\[2ex]
    $D^{(H)}_{\lambda(4)}$ & $\frac{1}{2}m(m-1)p(p+1)D_e$ &$\frac{1}{2}m(m-1)p(p+1)D_o$\\[2ex]
    $D^{(H)}_{\lambda(3,1)}$ &$-\frac{1}{2}m(m-1)p(p+1)D_e$ &$-\frac{1}{2}m(m+1)p(p+1)D_o$\\[2ex]
    $D^{(H)}_{\lambda(2^2)}$ &$m^2p^2D_e$ &$m(m+1)p^2D_o$\\[2ex]
    $D^{(H)}_{\lambda(2,1^2)}$ &$-\frac{1}{2}m(m+1)p(p-1)D_e$ &$-\frac{1}{2}m(m+1)p(p-1)D_o$\\[2ex]
    $D^{(H)}_{\lambda(1^4)}$ &$\frac{1}{2}m(m+1)p(p-1)D_e$ &$\frac{1}{2}(m+2)(m+1)p(p-1)D_o$ \\
    \hline\hline
    \end{tabular}
    \caption{The values of determinant $D^{(H)}_{\lambda\nu}$ for $\lambda=(N^{2p})$. Determinants $D_e$ and $D_o$ are given in \eqref{eq:De_Do}.}
    \label{table:D_lambda_nu}
\end{table}

 Therefore,
 \beq
 \begin{split}
 \mathbb{E}^{(H)}_N\left[\det(t-M)^{2p}\right] =& \sum_{\nu\subseteq\lambda}\left(-\frac{1}{2N}\right)^{\frac{|\lambda|-|\nu|}{2}}\frac{\text{dim}\, V_\nu}{|\nu|!}t^{|\nu|}\,\text{poly}_{\frac{|\nu|}{2}}(N,p)\\
 &\times
 \begin{cases}
 C_\lambda(2p) D_e,\quad \text{$N$ even},\\
 C_\lambda(2p)D_o,\quad \text{$N$ odd},
 \end{cases}
 \end{split}
 \eeq
where $\text{poly}_{j}(N,p)$ denotes a polynomial of degree $j$ in variables $N$, $p$, and the explicit expressions are given in Table.~\ref{table:D_lambda_nu} for $j\leq 4$. By referring to \eqref{eq:C_lambda_D}, it is remarkable to see that the universal constant $\gamma_p$ is a factor of the moments for any finite $N$. The first few terms in the moments of characteristic polynomials are
\beq\label{eq:moments_t_exp}
\begin{split}
\mathbb{E}^{(H)}_N[\det(t-M)^{2p}]&=\left(-\frac{1}{2N}\right)^{Np}C_\lambda(2p)D_e\\
&\quad\times \left[1 + \left(\frac{2^2N^2}{4!}\right)Np\,t^4 + \left(\frac{2^3N^3}{6!}\right)2Np(2p-N)t^6 \right.\\
&\qquad\left. +\left(\frac{2^4N^4}{8!}\right)Np (4 N^2 - 17 Np + 16 p^2+2)t^8 \right.\\
&\qquad\left.+ O(t^{10}) \right],\qquad\qquad\qquad\qquad\qquad\qquad\qquad\qquad \text{$N$ even},\\
\mathbb{E}^{(H)}_N[\det(t-M)^{2p}]&=\left(-\frac{1}{2N}\right)^{Np}C_\lambda(2p)D_o\left[1 +\left(\frac{2N}{2!}\right)pt^2 +\left(\frac{2^2N^2}{4!}\right)(p^2-Np)\,t^4\right.\\
&\quad\left. + \left(\frac{2^3N^3}{6!}\right)p(2N^2-3Np+p^2)t^6 \right.\\
&\quad\left. +\left(\frac{2^4N^4}{8!}\right)p (- 4 N^3+ 15 N^2p  - 6 N p^2 -2N+ p^3-4p)t^8\right.\\
&\quad\left. +\,\, O(t^{10}) \right],\quad\qquad\qquad\qquad\qquad\qquad\qquad\qquad\qquad \text{$N$ odd}.
\end{split}
\eeq

Up to a factor of $(-1)^p$, both $C_\lambda D_e$ and $C_\lambda D_o$ have the same leading term,
\beq
e^{-Np}(2N)^{Np+p^2}\gamma_p,
\eeq
but they differ at sub-leading order as shown in \eqref{eq:asymp_ratio_factorials}. In App.~\ref{app:asymp_ratio_factorials}, we give the asymptotic expansion of $C_\lambda(N)D_e$ and $C_\lambda(N)D_o$ up to $O(N^{-6})$. Note that the coefficients of $t^{2j}$ in \eqref{eq:moments_t_exp} are polynomials in $N$, and both $C_\lambda D_e$ and $C_\lambda D_o$ have an expansion in $1/N$. Therefore for higher values of $j$, more sub-leading terms in the expansion of $C_\lambda(N)D_e$ and $C_\lambda(N)D_o$  are required to compute the correct coefficients of $t^{2j}$. But finding the exact asymptotic expansion of $C_\lambda D_e$ and $C_\lambda D_o$ is far from trivial as it involves a sequence of ratios of factorials, whose asymptotics is only known via recurrence relations.

In the next section, we focus on the second moment and show that we recover the semi-circle law only after averaging over even and odd matrix dimensional contributions. 

\subsubsection{Second moment}\label{sec:second moment}
The correlations of characteristic polynomials are connected to the correlation functions of random matrices \cite{Mehta2004,Forrester2010,Mezzadri2005}. In particular, 
\beq\label{eq:one-point to second moment}
R^{(N)}_1(t) = \frac{N!}{(N-1)!}\frac{\mathscr{Z}^{(H)}_{N-1}}{\mathscr{Z}^{(H)}_N}\exp\left(-\frac{N t^2}{2}\right)\mathbb{E}^{(H)}_{N-1}\left[\det(t-M)^{2}\right],
\eeq
where $R^{(N)}_1$ is the one-point density of eigenvalues of matrix size $N$. As the second moment of the characteristic polynomial is related to the density of states, it is natural to expect the semi-circle law in the limit $N\rightarrow\infty$ as given in \eqref{eq:BH}.

Re-writing \eqref{eq:BH} for $p=1$,
\beq
\lim_{N\rightarrow\infty}\frac{1}{2N}e^{N\left(1-\frac{ t^2}{2}\right)}\mathbb{E}^{(H)}_{N}\left[\det(t-M)^{2}\right] = \pi\rho_{sc}(t),
\eeq
which as an expansion in $t$ reads 
\beq\label{eq:p1case}
\lim_{N\rightarrow\infty}\frac{1}{2N}e^{N\left(1-\frac{ t^2}{2}\right)}\mathbb{E}^{(H)}_{N}\left[\det(t-M)^{2}\right] = 1 - \frac{1}{8}t^2 - \frac{1}{128}t^4 - \frac{1}{1024}t^6 + O(t^8).
\eeq
We now show that for $p=1$, starting with \eqref{eq:char_mom} we arrive at \eqref{eq:p1case}. Inserting the asymptotics of $C_\lambda D_e$ and $C_\lambda D_o$ in \eqref{eq:moments_t_exp}, one obtains
\beq
\begin{split}
&e^{-\frac{N t^2}{2}}\mathbb{E}^{(H)}_{N}\left[\det(t-M)^{2}\right]\\
&= 2Ne^{-N}\bigg[1+
\left(-\frac{5}{12}-\frac{1}{2}N\right)t^2+\left(-\frac{811}{77760}+\frac{17}{216}N +\frac{19}{72}N^2+\frac{1}{6}N^3\right)t^4\\
&\qquad+\left(-\frac{640879}{587865600}+\frac{799}{1749600}N-\frac{3667}{291600}N^2-\frac{323}{6480}N^3-\frac{31}{540}N^4-\frac{1}{45}N^5\right)t^6\\
&\qquad+ O(t^8)\bigg],\quad\text{$N$ even},\\
&e^{-\frac{N t^2}{2}}\mathbb{E}^{(H)}_{N}\left[\det(t-M)^{2}\right]\\
&= 2Ne^{-N}\bigg[1+\left(\frac{1}{6}+\frac{1}{2}N\right)t^2+\left(-\frac{101}{19440}-\frac{17}{216}N-\frac{19}{72}N^2-\frac{1}{6}N^3\right)t^4\\
&\quad +\left(-\frac{15853}{18370800}-\frac{799}{1749600}N+\frac{3667}{291600}N^2+\frac{323}{6480}N^3+\frac{31}{540}N^4+\frac{1}{45}N^5\right)t^6\\
&\qquad + O(t^8)\bigg],\quad\text{$N$ odd}.\\
\end{split}
\eeq
Treating the above expansions as a formal series in $N$ and taking their average gives \eqref{eq:p1case}. 
In App.~\ref{app:more about second moment}, it is shown that the average over even and odd $N$ coincides with the semi-circle law up to $O(t^{10})$. Also, a general expression for the coefficient of $t^{2j}$ in \eqref{eq:char_mom} is given for $p=1$.

\subsubsection{Higher moments}
For higher moments, the correlations of characteristic polynomials are related to the correlation functions of eigenvalues as 
\beq
R^{(N)}_p(t_1,\dots,t_p) = \frac{N!}{(N-p)!}\frac{\mathscr{Z}^{(H)}_{N-p}}{\mathscr{Z}^{(H)}_N}\exp\left(-\frac{N}{2}\sum_{j=1}^p t_j^2\right)\Delta^2(t_1,\dots,t_p)\mathbb{E}^{(H)}_{N-p}\left[\prod_{j=1}^p\det(t_j-M)^{2}\right],
\eeq
where $R^{(N)}_p(t_1,\dots,t_p)$ denotes a $p-$point correlation function of a GUE matrix of size $N$. The correlations of characteristic polynomials of matrices of size $N-p$ are related to the correlation functions of eigenvalues of matrices of size $N$. The Dyson sine-kernel for the $p-$point correlation function and \eqref{eq:BH} for the moments of characteristic polynomials are recovered in the Dyson limit: $t_i-t_j\rightarrow 0$, $N\rightarrow\infty$ and $N(t_i-t_j)$ is kept finite when $|t_j|<2$, $j=1,\dots,p$.

In terms of the Schur polynomials, $\lambda=(N^{2p})$,
\beq\label{eq:correlations charpoly}
\mathbb{E}^{(H)}_N\left[\prod_{j=1}^{2p}\det(t_j-M)\right]=C_{\lambda}(2p) \sum_{\nu\subseteq\lambda}\left(-\frac{1}{2N}\right)^{\frac{|\lambda|-|\nu|}{2}}\frac{1}{C_\nu(2p)}D^{(H)}_{\lambda \nu}S_\nu(t_1,\dots,t_{2p}).
\eeq
Computing the asymptotics of moments of characteristic polynomials in the Dyson limit using \eqref{eq:correlations charpoly} is highly non-trivial. Instead, we fix $t_j=t$, $j=1,\dots, 2p$, and give an expansion of the moments as a function of $t$ in the large $N$ limit.

As $N\rightarrow\infty$, up to $O(t^2)$,
\beq
\begin{split}
&\mathbb{E}^{(H)}_N\left[\det(t-M)^{2p}\right]  = (2N)^{p^2}e^{-Np}\gamma_p\left[1 + O(t^4)\right],\qquad\qquad\qquad\qquad\qquad\qquad\,\, N\,\,\text{even}\\
&\mathbb{E}^{(H)}_N\left[\det(t-M)^{2p}\right]  = (2N)^{p^2}e^{-Np}\gamma_p\left[1 + t^2\left(Np+\frac{p^2}{3}(2p^2-1) \right) + O(t^4)\right],\quad N\,\,\text{odd}.
\end{split}
\eeq
Note that the coefficient of $t^2$ is identically zero for even $N$, where as for odd $N$ it is a polynomial in $N$ and $p$. Treating the above expansions as a formal series in $N$ and taking their average gives
\beq\label{eq:upto t_sqr}
(2N)^{p^2}e^{-Np}\gamma_p\left(1 + \frac{Npt^2}{2} \right)\left(1 - \frac{p^2t^2}{8}\right)\left(1 + \frac{p}{12N}(8p^2-1)\right).
\eeq
By comparing with \eqref{eq:BH}, the terms in the first and second parenthesis of \eqref{eq:upto t_sqr} are the expansions of $e^{\frac{Npt^2}{2}}$ and $\pi\rho_{sc}(t)$ up to $O(t^2)$ respectively. The last factor in \eqref{eq:upto t_sqr} is sub-leading. Thus at $O(t^2)$, moments of characteristic polynomials in the Dyson limit and in the limit $t\rightarrow 0$ and $N\rightarrow\infty$ coincide.

Similarly, as $N\rightarrow\infty$, up to $O(t^4)$,
\beq
\begin{split}
\mathbb{E}^{(H)}_N\left[\det(t-M)^{2p}\right]  &= (2N)^{p^2}e^{-Np}\gamma_p\left[1 + t^4\frac{N^3p}{6}\Big(1+\frac{p}{6N}(4p^2+1) \right.\\
&\qquad \left. + \frac{p^2}{72N^2}(16p^4-16p^2-11)+\frac{p}{6480N^3}(320p^8-1200p^6 \right.\\
&\qquad \left. +708p^4+1265p^2-756)\Big) + O(t^{6})\right],\hspace{40pt} \text{$N$ even},\\
\mathbb{E}^{(H)}_N\left[\det(t-M)^{2p}\right]  &= (2N)^{p^2}e^{-Np}\gamma_p\left[1 + t^2\left(Np+\frac{p^2}{3}(2p^2-1) \right)\right.\\
&\qquad\left. + t^4\frac{N^3p}{6}\Big(-1 - \frac{2p}{3N}(p^2-2) -\frac{p^2}{18N^2}(4p^4-22p^2+13)\right.\\
&\qquad\left. -\frac{p}{405N^3}(20p^8-210p^6+483p^4-385p^2+54)\Big) \right.\\
&\qquad\left. +O(t^6)\right],\hspace{175pt} N\,\,\text{odd}.
\end{split}
\eeq
Taking average of the above series and factorising gives
\beq
\begin{split}
&(2N)^{p^2}e^{-Np}\gamma_p\left(1 +  \frac{Npt^2}{2} + \frac{N^2p^2t^4}{8} + O(t^6)\right)\left(1 - \frac{p^2t^2}{8} + \frac{t^4}{128}p^2(p^2-2) + O(t^6)\right)\\
&\times\left[1 + \frac{1}{N}\left(\frac{p}{12}(8p^2-1) + \frac{pt^2}{96}(13p^2-1)+O(t^4)\right) + \frac{1}{N^2}\left(\frac{p^2}{144}(32p^4 - 56p^2 + 17) + O(t^2)\right)\right],
\end{split}
\eeq
where the first two brackets correspond to the expansion of $e^{\frac{Npt^2}{2}}$ and $\pi\rho_{sc}(t)$, respectively, up to $O(t^4)$, and the last factor is sub-leading.   Thus, asymptotics calculated by letting first $t\rightarrow 0$ and then $N\rightarrow \infty$ coincides with that of Dyson limit asymptotics up to $O(t^4)$. For higher orders in $t$, mismatch between the two limits starts to appear.

\section{Secular coefficients}\label{sec:secular coeff}
Consider a matrix $M$ of size $N$. Its characteristic polynomial can be expanded as
\beq
\det(t-M)=\prod_{j=1}^N(t-x_j) = \sum_{j=0}^N(-1)^{j}\secu_j(M)t^{N-j},
\eeq
where $\secu_j$ is the $j^{th}$ secular coefficient of the characteristic polynomial. We have
\beq 
\secu_1(M) = \Tr M,\quad \secu_N(M) = \det(M).
\eeq
These secular coefficients are nothing but the elementary symmetric polynomials $e_j$ defined as
\beq
e_j(x_1,\dots,x_N) = \sum_{1\leq k_1<k_2<\dots <k_j\leq N}x_{k_1}x_{k_2}\dots x_{k_j}
\eeq
for $j\leq N$ and $e_j=0$ for $j>N$.

The correlations of secular coefficients and their connections to combinatorics have been studied in the past \cite{Forrester2006,Diaconis2004}. For example, the joint moments of secular coefficients of the unitary group are connected to the enumeration of magic squares: matrices with positive entries with prescribed row and column sum. In a similar way, the joint moments of secular coefficients of Hermitian ensembles, such as the GUE, are connected to matching polynomials of closed graphs. In this section, we compute these correlations and indicate their combinatorial properties.

\noindent\textit{Gaussian ensemble:} Elementary symmetric polynomials can be expanded in terms of multivariate Hermite polynomials as
\beq\label{eq:er to mulher}
e_r  = \sum_{j=0}^{\lfloor\frac{r}{2}\rfloor}\Psi^{(H)}_{(1^r)(1^{r-2j})}\mathscr{H}_{(1^{r-2j})},
\eeq
where
\beq\label{eq:psi her}
\Psi^{(H)}_{(1^r)(1^{r-2j})} = (-1)^j\frac{(N-r+2j)!}{(2N)^jj!(N-r)!}.
\eeq
Equivalently, we have
\beq
\begin{split}
e_{2r} = \sum_{j=0}^r\Psi^{(H)}_{(1^{2r})(1^{2j})}\mathscr{H}_{(1^{2j})},\quad e_{2r+1} = \sum_{j=0}^r\Psi^{(H)}_{(1^{2r+1})(1^{2j+1})}\mathscr{H}_{(1^{2j+1})},
\end{split}
\eeq
with
\beq
\begin{split}
\Psi^{(H)}_{(1^{2r})(1^{2j})} &= (-1)^{r-j}\frac{1}{(2N)^{r-j}(r-j)!}\frac{(N-2j)!}{(N-2r)!},\\
\Psi^{(H)}_{(1^{2r+1})(1^{2j+1})} &= (-1)^{r-j}\frac{1}{(2N)^{r-j}(r-j)!}\frac{(N-2j-1)!}{(N-2r-1)!}.
\end{split}
\eeq 
Because of the orthogonality of the $\mathcal{H}_\mu$, 
\beq
\mathbb{E}^{(H)}_N[\secu_r] = \mathbb{E}^{(H)}_N[e_r] = 
\begin{cases}
(-1)^{\frac{r}{2}}\frac{1}{(2N)^{\frac{r}{2}}\frac{r}{2}!}\frac{N!}{(N-r)!},\quad \text{if $r$ is even},\\
0, \hspace{8.8em}\text{if $r$ is odd}.
\end{cases}
\eeq
These expectations are nothing but the coefficients of Hermite polynomial of degree $N$. Thus,
\beq\label{eq:first moment gue}
\begin{split}
\mathbb{E}^{(H)}_N[\det(t-M)] &= \sum_{j=0}^{\lfloor\frac{N}{2}\rfloor}\mathbb{E}^{(H)}_N[\secu_{2j}(M)]t^{N-2j}=h_N(t),
\end{split}
\eeq
which coincides with \eqref{eq:char poly corre her rescaled} for $p=1$. The expectation $|N^j\mathbb{E}^{(H)}_N[\secu_{2j}(M)]|$ is equal to the number of $2j$ matchings in the complete graphs \cite{Diaconis2004,Forrester2006}.

By using \eqref{eq:er to mulher}, the second moment of the secular coefficient can also be computed. Similar to the univariate case,  multivariate Hermite polynomials $\mathscr{H}_\lambda$ corresponding to even and odd $|\lambda|$ do not mix. Hence, we obtain
\beq
\mathbb{E}^{(H)}_N[\secu_{2j}(M)\secu_{2k+1}(M)] = 0,
\eeq
and
\beq\begin{split}
\mathbb{E}^{(H)}_N[\secu_{2r}(M)\secu_{2s}(M)] &= \sum_{j=0}^r\sum_{k=0}^s\Psi^{(H)}_{(1^{2r})(1^{2j})}\Psi^{(H)}_{(1^{2s})(1^{2k})}\mathbb{E}^{(H)}_N[\mathscr{H}_{(1^{2j})}\mathscr{H}_{(1^{2k})}]\\
&= \sum_{j=0}^{\min(r,s)}\frac{1}{N^{2j}}\Psi^{(H)}_{(1^{2r})(1^{2j})}\Psi^{(H)}_{(1^{2s})(1^{2j})}C_{(1^{2j})}(N)\\
&=\left(-\frac{1}{2N}\right)^{r+s}\, \sum_{j=0}^{\min(r,s)}\frac{2^{2j}}{(r-j)!(s-j)!}\frac{N!(N-2j)!}{(N-2r)!(N-2s)!}.
\end{split}
\eeq
Similarly, we write
\beq
\mathbb{E}^{(H)}_N[\secu_{2r+1}(M)\secu_{2s+1}(M)] = \left(-\frac{1}{2N}\right)^{r+s}\,\sum_{j=0}^{\min(r,s)}\frac{2^{2j}}{(r-j)!(s-j)!}\frac{(N-1)!(N-2j-1)!}{(N-2r-1)!(N-2s-1)!}.
\eeq
Computing higher order correlations requires evaluating integrals involving a sequence of multivariate Hermite polynomials. Busbridge \cite{Busbridge1939,Busbridge1948} calculated these integrals for the univariate case, but the results are still unknown for the multivariate generalisation. Instead, we take a different approach by first expressing the product $\prod_j\secu_j(M))^{b_j}$ in terms of the $\mathscr{H}_\mu$ and then using orthogonality for the $\mathscr{H}_\mu$.

\begin{proposition}
Consider a partition $\lambda = (\lambda_1,\dots ,\lambda_l)$. We have
\beq
\mathbb{E}^{(H)}_N\big[\prod_{j=1}^l\secu_{\lambda_j}(M)\big] = 
\begin{cases}
\sum_\mu \frac{1}{(2N)^{\frac{\mu}{2}}\frac{|\mu|}{2}!}K_{\lambda^\prime\mu}\chi^{\mu}_{(2^{|\mu|/2})}C_\mu(N),\quad \text{if $|\lambda|$ is even},\\
0,\hspace{160pt} \text{otherwise}.
\end{cases}
\eeq
Here $K_{\lambda\mu}$ are Kostka numbers\footnote{The Kostka numbers are non-negative integers that count the number of semi-standard Young tableau of shape $\lambda$ and weight $\mu$.} and $\chi^\mu_\nu$ is the character of the symmetric group.
\end{proposition}
\begin{proof}
For a partition $\lambda$, denote 
\beq
e_\lambda = e_{\lambda_1}e_{\lambda_2}\dots. 
\eeq
Elementary symmetric polynomials $e_\lambda$ can be expanded in Schur basis as follows:
\beq
e_\lambda = \sum_{\mu}K_{\lambda^\prime \mu}S_\mu,
\eeq
where $K_{\lambda\mu}$ are the Kostka numbers \cite{Macdonald1998} and $\mu$ is a partition of $|\lambda|$. Using \eqref{eq:schur to mulher},
\beq
e_\lambda = \sum_{\mu\vdash |\lambda|}\sum_{\nu\subseteq\mu}K_{\lambda^\prime\mu}\Psi^{(H)}_{\mu\nu}\mathscr{H}_\nu.
\eeq
When $|\lambda|$ is odd, $\mathbb{E}^{(H)}_N[e_\lambda]=0$ due to the orthogonality of multivariate Hermite polynomials. When $|\lambda|$ is even, 
\beq
\begin{split}
\mathbb{E}^{(H)}_N[e_\lambda] &= \mathbb{E}^{(H)}_N\big[\prod_{j=1}^l\secu_{\lambda_j}(M)\big]\\
 &= \mathbb{E}^{(H)}_N\big[\sum_{\mu}\sum_\nu K_{\lambda^\prime\mu}\Psi^{(H)}_{\mu  \nu}\mathscr{H}_\nu\big]\\
 &=K_{\lambda^\prime\mu}\Psi^{(H)}_{\mu 0}.
\end{split}
\eeq
It can be shown that \cite{Jonnadula2020}
\beq
\Psi^{(H)}_{\mu 0} = \frac{1}{(2N)^{\frac{\mu}{2}}\frac{|\mu|}{2}!}\chi^{\mu}_{(2^{|\mu|/2})}C_\mu(N).
\eeq
Putting everything together completes the proof.
\end{proof}

\noindent\textit{Laguerre ensemble:} All the calculations discussed for the Gaussian ensemble can be extended to the Laguerre and the Jacobi ensembles. 

The polynomials $e_r$ can be expanded as 
\beq\label{eq:er to mullag}
e_r  = \sum_{j=0}^{r}\Psi^{(L)}_{(1^r)(1^{j})}\mathscr{L}^{(\gamma)}_{(1^{j})},
\eeq
where
\beq\label{eq:psi lag}
\Psi^{(L)}_{(1^r)(1^{j})} = \frac{1}{(2N)^{r-j}}\frac{1}{(r-j)!}\frac{(N-j)!}{(N-r)!}\frac{\Gamma(N-j+\gamma +1)}{\Gamma(N-r+\gamma + 1)}.
\eeq
By using \eqref{eq:ortho mullag} we arrive at
\beq
\mathbb{E}^{(L)}_N[\secu_r] = \mathbb{E}^{(L)}_N[e_r] = \frac{1}{(2N)^r}\frac{1}{r!}\frac{N!}{(N-r)!}\frac{\Gamma(N+\gamma +1)}{\Gamma(N-r+\gamma + 1)},
\eeq
which are the absolute values of the coefficients of the Laguerre polynomial of degree $N$. For the characteristic polynomial, we have
\beq\label{eq:first moment lue}
\mathbb{E}^{(L)}_N[\det(t-M)] = \sum_{j=0}^N(-1)^j\mathbb{E}^{(L)}_N[\secu_j(M)]t^{N-j} = l_N^{(\gamma)}(t).
\eeq
The correlations of secular coefficients can be computed similar to the Gaussian case.
\begin{proposition}
Let  $\lambda = (\lambda_1,\dots ,\lambda_l)$, we have
\beq
\mathbb{E}^{(L)}_N\big[\prod_{j=1}^l\secu_{\lambda_j}(M)\big] =\sum_{\mu\vdash |\lambda|}\frac{1}{(2N)^{|\lambda|}}\frac{G_\mu(N,\gamma)G_\mu(N,0)}{G_0(N,\gamma)G_0(N,0)}\frac{\chi^\mu_{(1^{|\mu|})}}{|\lambda|!} K_{\lambda^\prime\mu} 
\eeq
\end{proposition}
\begin{proof}
The proof is similar to the Gaussian case. By writing 
\beq
e_\lambda = \sum_\mu\sum_{\nu\subseteq |\lambda|} K_{\lambda^\prime\mu}\Psi^{(L)}_{\mu\nu}\mathscr{L}^{(\gamma)}_\nu,
\eeq
and using \eqref{eq:ortho mullag} along with the result \cite{Jonnadula2020}
\beq
\Psi^{(L)}_{\mu 0} = \frac{1}{(2N)^{|\mu|}}\frac{G_\mu(N,\gamma)G_\mu(N,0)}{G_0(N,\gamma)G_0(N,0)}\frac{\chi^\mu_{(1^{|\mu|})}}{|\mu|!}
\eeq
proves the proposition.
\end{proof}

\noindent{\textit{Jacobi ensemble.}} The $e_r$ can be expanded as
\beq
e_r = \sum_{j=0}^{r}\Psi^{(J)}_{(1^r)(1^{j})}\mathscr{J}^{(\gamma_1,\gamma_2)}_{(1^{j})},
\eeq
where $\Psi^{(J)}_{\lambda\nu}$ is given in \eqref{eq: schur to mul jac coef}. The expected values of the $e_r$ are related to the coefficients of the Jacobi polynomial of degree $N$.
\beq\label{eq:first moment jue}
\mathbb{E}^{(J)}_N[\det(t-M)] = \sum_{j=0}^{N}(-1)^j\mathbb{E}^{(J)}_N[\secu_{j}(M)]t^{N-j}=j^{(\gamma_1,\gamma_2)}_N(t)
\eeq

\textit{Acknowledgements}. FM is grateful for support from the  University Research Fellowship of the University of Bristol. JPK is pleased to acknowledge support from a Royal Society Wolfson Research Merit Award and ERC Advanced Grant 740900 (LogCorRM).

\appendix
\section*{Appendix}

\section{Asymptotics of ratio of factorials}\label{app:asymp_ratio_factorials}
The asymptotics of the ratio of factorials can be computed as follows. First we look at $C_\lambda(2p)D_{e}$ with $\lambda=(2m,\dots,2m)$. Consider
\beq
\begin{split}
\frac{(2m+j)!(2m+p+j)!}{(m+j)!^2}&=(2m)^p\frac{(2m+j)!^2}{(m+j)!^2}\prod_{a=1}^p\left(1+\frac{j+a}{2m}\right).
\end{split}
\eeq
Now, one can see that
\beq
\frac{(2m+j)!}{(m+j)!} = 2^{j+1}\frac{\Gamma(2m)}{\Gamma(m)}\prod_{a=0}^j\frac{1+\frac{a}{2m}}{1+\frac{a}{m}}.
\eeq
Using the duplication formula for the Gamma functions
\beq\label{eq:duplication_formula}
\Gamma(z)\Gamma\left(z+{\frac{1}{2}}\right)=2^{1-2z}\sqrt{\pi}\,\Gamma(2z)
\eeq
and Stirling's series
\beq\label{eq:stirlings_series}
\Gamma(z+h)\sim \sqrt{2\pi}e^{-z}z^{z+h-\frac{1}{2}}\prod_{j=2}^\infty \exp\left({\frac{(-1)^jB_j(h)}{j(j-1)z^{j-1}}}\right),\quad z\rightarrow\infty,
\eeq
the asymptotic expansion for the ratio of Gamma functions can be found. Here $B_j$ is the Bernoulli polynomial of degree $j$. Combining all the formulae, up to first order correction,
\beq
C_{((2m)^{2p})}(2p) D_e\sim e^{-2mp}2^{4mp+2p^2}m^{2mp+p^2}\left(\prod_{j=0}^{p-1}\frac{j!}{(p+j)!}\right)\left[1 + \frac{p}{12m}(4p^2+1)+O(m^{-2})\right].
\eeq
Similarly for the case $C_\lambda(2p)D_o$, we obtain
\beq
\begin{split}
\frac{(2m+1+p+j)!(2m+1+j)!}{(m+j)!^2}= (2m+1)^p\frac{(2m+1+j)!^2}{(m+j)!^2}\prod_{a=1}^p\left(1 + \frac{j+a}{2m+1}\right).
\end{split} 
\eeq
Let $z=m+\frac{1}{2}$, then
\beq
\frac{(2m+1+j)!}{(m+j)!} = \frac{\Gamma(2z+j+1)}{\Gamma(z+\frac{1}{2}+j)} = 2^{j+1}z\frac{\Gamma(2z)}{\Gamma(z+\frac{1}{2})}\prod_{a=1}^j\frac{1+\frac{a}{2z}}{1+\frac{2a-1}{2z}},
\eeq
and
\beq
\frac{m!}{(m+p)!} = \frac{\Gamma(z+\frac{1}{2})}{\Gamma(z+p+\frac{1}{2})} = \frac{1}{z^p} \prod_{a=1}^p\frac{1}{1+\frac{2a-1}{2z}}
\eeq
Combining the above formulae and using \eqref{eq:duplication_formula} and \eqref{eq:stirlings_series},
\beq
\begin{split}
C_{((2m+1)^{2p})}D_o\equiv C_{((2z)^{2p})}D_o\sim & (-1)^pe^{-2zp}z^{p^2+2pz}2^{2p^2+4pz}\left(\prod_{j=0}^{p-1}\frac{j!}{(p+j)!}\right)\\&\quad\times \left[1 + \frac{p}{6z}(2p^2-1)+O(z^{-2})\right]
\end{split}
\eeq
Higher order corrections can also be calculated with some effort or using any commercial software like Mathematica. Writing in terms of the matrix size $N$, as $N\rightarrow\infty$, we have
\beq
\begin{split}
&C_{(N^{2p})}D_e \sim e^{-Np}(2N)^{Np+p^2}\left(\prod_{j=0}^{p-1}\frac{j!}{(p+j)!}\right)\left[1+\frac{p}{6N}(4p^2+1) + \frac{p^2}{72N^2}(16p^4-16p^2-11) \right.\\
&\qquad\qquad\qquad\left. +\frac{p}{6480N^3}(320p^8-1200p^6+708p^4+1265p^2-756) \right.\\
&\qquad\qquad\qquad\left. +\frac{p^2}{155520N^4} (  1280 p^{10}- 10240 p^8+ 25248 p^6 - 6400 p^4- 56371 p^2 +51408)\right.\\
&\qquad\qquad\qquad \left. +\frac{p}{6531840N^5}\Big(7168 p^{14}- 98560 p^{12}  +499072 p^{10}   - 982688 p^8 - 399844 p^6\right.\\
&\qquad\qquad\qquad\qquad\left. + 4606735 p^4- 5598936 p^2 +1607040 \Big)\right.\\
&\qquad\qquad\qquad\left. + \frac{p^2}{1175731200N^6}\Big(143360 p^{16} -3010560 p^{14} + 25294080 p^{12} - 103093760 p^{10} \right.\\
&\qquad\qquad\qquad\qquad\left.+ 158864016 p^8  +298943760 p^6- 1697420809 p^4  + 2663679600 p^2 -1390123296\Big)\right.\\
&\qquad\qquad\qquad\left. +O\left(\frac{1}{N^{7}}\right)\right], \qquad\qquad\qquad\qquad\qquad\qquad\qquad\qquad\qquad\qquad\qquad\qquad N\,\, \text{even}.
\end{split}
\eeq

\beq
\begin{split}
&C_{(N^{2p})}D_o\sim (-1)^pe^{-Np}(2N)^{Np+p^2}\left(\prod_{j=0}^{p-1}\frac{j!}{(p+j)!}\right)\left[1 + \frac{p}{3N}(2p^2-1) + \frac{p^2}{18N^2}(4p^4-10p^2+7)\right.\\
&\left. \qquad\qquad \qquad + \frac{p}{810N^3}(40p^8-240p^6+516p^4-455p^2+108)\right.\\
&\qquad\qquad\qquad\left.+\frac{p^2}{9720N^4}(80 p^{10}- 880 p^8+ 3828 p^6 - 8356 p^4 + 9509 p^2-4320)\right.\\
&\qquad\qquad\qquad \left. +  \frac{p}{204120N^5}\Big(224\,p^{14} - 3920\, p^{12} +28616\,p^{10}  - 113428\,p^8 + 266818\,p^6\right.\\
&\qquad\qquad\qquad\qquad\left. - 372127\,p^4 + 255528\,p^2 -51840 \Big)\right.\\
&\qquad\qquad\qquad\left.  +\frac{p^2}{18370800N^6}\Big(2240\,p^{16} - 57120\,p^{14}+ 628320\,p^{12} - 3919160\,p^{10}  + 15363624\,p^8\right.\\
&\qquad\qquad\qquad\qquad\left. - 39481170\,p^6 + 65605589\,p^4  - 62864640\,p^2 +25046496 \Big)\right.\\
  &\qquad\qquad\qquad\left. + O\left(\frac{1}{N^7}\right)\right], \qquad\qquad\qquad\qquad\qquad\qquad\qquad\qquad\qquad\qquad\qquad\qquad N\,\, \text{odd}.
\end{split}
\eeq

\section{More on the second moment}
\label{app:more about second moment}
Fix $\lambda=(N,N)$. The second moment of the characteristic polynomial is given by
\beq\label{eq:appc.0}
\mathbb{E}^{(H)}_{N}\left[\det(t-M)^{2}\right] = \left(\frac{-1}{2N}\right)^NC_\lambda(2)\sum_{\nu\subseteq \lambda}\frac{1}{|\nu|!}(-2N)^{\frac{|\nu|}{2}}D^{(H)}_{\lambda\nu}\text{dim}\,V_\nu\,t^{|\nu|}.
\eeq
Let $\nu=(\nu_1,\nu_2)\subseteq \lambda$. Since $|\nu|$ is even, either both $\nu_1$, $\nu_2$ are even or both of them are odd. 
For $N=2m$, $m\in\mathbb{N}$, 
\beq
D^{(H)}_{\lambda\nu} = 
\begin{cases}
\frac{1}{\left(m-\frac{\nu_1}{2}\right)!\left(m-\frac{\nu_2}{2}\right)!},\qquad\qquad\,\, \text{$\nu_1$, $\nu_2$ are even},\\
-\frac{1}{\left(m-\frac{\nu_1+1}{2}\right)!\left(m-\frac{\nu_2-1}{2}\right)!},\quad\quad\text{$\nu_1$, $\nu_2$ are odd}.
\end{cases}
\eeq
Therefore,
\beq\label{eq:appc.1}
C_\lambda(2)D^{(H)}_{\lambda\nu} = 
(2m)!(2m+1)!\begin{cases}
\frac{1}{\left(m-\frac{\nu_1}{2}\right)!}\frac{1}{\left(m-\frac{\nu_2}{2}\right)!},\qquad\quad\,\, \text{$\nu_1$, $\nu_2$ are even},\\
-\frac{1}{\left(m-\frac{\nu_1+1}{2}\right)!}\frac{1}{\left(m-\frac{\nu_2-1}{2}\right)!},\quad\,\text{$\nu_1$, $\nu_2$ are odd}.
\end{cases}
\eeq
Similarly, for $N=2m+1$, $m\in\mathbb{N}$, we have
\beq
D^{(H)}_{\lambda\nu} = 
\begin{cases}
-\frac{1}{\left(m-\frac{\nu_1}{2}\right)!\left(m-\frac{\nu_2-2}{2}\right)!},\qquad\quad \text{$\nu_1$, $\nu_2$ are even},\\
\frac{1}{\left(m-\frac{\nu_1-1}{2}\right)!\left(m-\frac{\nu_2-1}{2}\right)!},\qquad\quad\text{$\nu_1$, $\nu_2$ are odd},
\end{cases}
\eeq
and
\beq\label{eq:appc.2}
C_\lambda(2)D^{(H)}_{\lambda\nu} = 
(2m+1)!(2m+2)!
\begin{cases}
-\frac{1}{\left(m-\frac{\nu_1}{2}\right)!}\frac{1}{\left(m-\frac{\nu_2-2}{2}\right)!},\quad \text{$\nu_1$, $\nu_2$ are even},\\
\frac{1}{\left(m-\frac{\nu_1-1}{2}\right)!}\frac{1}{\left(m-\frac{\nu_2-1}{2}\right)!},\quad\text{$\nu_1$, $\nu_2$ are odd}.
\end{cases}
\eeq
For a partition of length 2, $\nu=(\nu_1,\nu_2)$,
\beq\label{eq:appc.3}
\frac{1}{|\nu|!}\text{dim}\,V_\nu = \frac{\nu_1-\nu_2+1}{(\nu_1+1)!\,\nu_2!}.
\eeq

Inserting \eqref{eq:appc.1}, \eqref{eq:appc.2}, \eqref{eq:appc.3} in \eqref{eq:appc.0}, and observing that $\nu$ runs over all partitions such that $0\leq|\nu|\leq 2N$ gives
\beq\label{eq:appc.4}
\begin{split}
&\mathbb{E}^{(H)}_{N}\left[\det(t-M)^{2}\right]\\
& = \left(-\frac{1}{2N}\right)^NC_\lambda(2)D^{(H)}_{\lambda 0}\sum_{k=0}^N(-2N)^kt^{2k}\\
&\quad\times\Big[\sum_{j=0}^{\lfloor{\frac{k-1}{2}}\rfloor}\left(\frac{2k+1-4j}{(2k+1-2j)!(2j)!}-\frac{2k-1-4j}{(2k-2j)!(2j+1)!}\right)\frac{\left(\frac{N}{2}\right)!^2}{(\frac{N}{2}-k+j)!(\frac{N}{2}-j)!}\\
&\qquad\qquad\quad + \frac{1}{k!(k+1)!}\frac{\left(\frac{N}{2}\right)!^2}{\left(\frac{N}{2}-\frac{k}{2}\right)!^2}\mathbbm{1}_{k= 0\,\text{mod}\,2}\Big].
\end{split}
\eeq
for $N$ even. Similarly, for $N$ odd, one gets
\beq\label{eq:appc.5}
\begin{split}
&\mathbb{E}^{(H)}_{N}\left[\det(t-M)^{2}\right]\\\
& = \left(-\frac{1}{2N}\right)^NC_\lambda(2)D^{(H)}_{\lambda 0}\sum_{k=0}^N(-2N)^kt^{2k}\\
&\times\bigg[\sum_{j=0}^{\lfloor{\frac{k-2}{2}}\rfloor}\left(-\frac{2k-1-4j}{(2k-2j)!(2j+1)!}+\frac{2k-3-4j}{(2k-2j-1)!(2j+2)!}\right)\frac{\left(\frac{N-1}{2}\right)!\,\left(\frac{N+1}{2}\right)!}{\left(\left(\frac{N+1}{2}\right)-k+j\right)!\left(\left(\frac{N-1}{2}\right)-j\right)!}\\
&\qquad\qquad\quad +\frac{1}{(2k)!}\frac{\left(\frac{N-1}{2}\right)!}{\left(\left(\frac{N-1}{2}\right)-k\right)!}- \frac{1}{k!(k+1)!}\frac{\left(\frac{N-1}{2}\right)!\,\left(\frac{N+1}{2}\right)!}{\left(\left(\frac{N-1}{2}\right)-\frac{k-1}{2}\right)!^2}\mathbbm{1}_{k= 0\,\text{mod}\,1}\bigg].
\end{split}
\eeq
We have 
\beq
C_\lambda(2)D^{(H)}_{\lambda 0} =
\begin{cases}
\frac{N!(N+1)!}{\left(\frac{N}{2}\right)!^2},\qquad\quad\,\,\,\, \text{N even},\\
-\frac{N!(N+1)!}{\left(\frac{N-1}{2}\right)!\left(\frac{N+1}{2}\right)!},\quad \text{$N$ odd}.
\end{cases}
\eeq
The asymptotics of the ratio of the factorials are already discussed in App.~\ref{app:asymp_ratio_factorials}. For the sake of completion, here we again give the result for $p=1$,
\beq
\begin{split}
C_\lambda(2)D^{(H)}_{\lambda 0}&\sim e^{-N}(2N)^{N+1}\left[1+\frac{5}{6\,N}-\frac{11}{72\,N^2}+\frac{337}{6480\,N^3}+\frac{985}{31104\,N^4}-\frac{360013}{6531840\,N^5}\right.\\
&\left.\qquad\qquad\qquad-\frac{46723609}{1175731200\,N^6} +\frac{224766221}{1410877440\,N^7}+\frac{41757020981}{338610585600\,N^8}\right.\\
&\left.\qquad\qquad\qquad-\frac{889926952101377}{1005673439232000\,N^9}+O(N^{-10})\right],\, \text{$N$ even},\\
C_\lambda(2)D^{(H)}_{\lambda 0}&\sim -e^{-N}(2N)^{N+1}\left[1+\frac{1}{3\,N}+\frac{1}{18\,N^2}-\frac{31}{810\,N^3}-\frac{139}{9720\,N^4}+\frac{9871}{204120\,N^5}\right.\\
&\left.\qquad\qquad\qquad+\frac{324179}{18370800\,N^6} -\frac{8225671}{55112400\,N^7}-\frac{69685339}{1322697600\,N^8}\right.\\
&\left.\qquad\qquad\qquad+\frac{1674981058019}{1964205936000\,N^9}+O(N^{-10})\right],\quad \text{$N$ odd}.
\end{split}
\eeq
 Substituting the above asymptotic series in 
\beq
(2N)^{-1}e^{N-\frac{Nt^2}{2}}\mathbb{E}^{(H)}_{N}\left[\det(t-M)^{2}\right]
\eeq
 and taking the average over $N$ even and odd gives
\beq\label{eq:appc.6}
\begin{split}
&\lim_{N\rightarrow\infty}\frac{1}{2N}e^{N}\exp\left(-\frac{N t^2}{2}\right)\mathbb{E}^{(H)}_{N}\left[\det(t-M)^{2}\right] \\
=& 1 - \frac{1}{8}t^2 - \frac{1}{128}t^4 - \frac{1}{1024}t^6-\frac{5}{32768}t^8-\frac{7}{262144}t^{10} + O(t^{12}).
\end{split}
\eeq
The R.H.S. in \eqref{eq:appc.6} coincides with $\pi\rho_{sc}(t)$ up to $O(t^{10})$.
 \bibliographystyle{abbrv}
\bibliography{mom_clt}{}

\begin{thebibliography}{10}

\bibitem{Akemann2020}
G.~Akemann, E.~Strahov, and T.~R. W{\"u}rfel.
\newblock Averages of products and ratios of characteristic polynomials in
  polynomial ensembles.
\newblock {\em Annales Henri Poincar{\'e}}, 21(12):3973--4002, 2020.

\bibitem{Andreev1995}
A.~Andreev and B.~Simons.
\newblock Correlators of spectral determinants in quantum chaos.
\newblock {\em Physical review letters}, 75(12):2304, 1995.

\bibitem{Baik2003}
J.~Baik, P.~Deift, and E.~Strahov.
\newblock {Products and ratios of characteristic polynomials of random
  Hermitian matrices}.
\newblock {\em {Journal of Mathematical Physics}}, 44(8):3657--3670, 2003.

\bibitem{Baker1997}
T.~H. Baker and P.~J. Forrester.
\newblock {The Calogero-Sutherland model and generalized classical
  polynomials}.
\newblock {\em {Communications in Mathematical Physics}}, 188(1):175--216,
  1997.

\bibitem{Baker1997calogero}
T.~H. Baker and P.~J. Forrester.
\newblock {The Calogero-Sutherland model and polynomials with prescribed
  symmetry}.
\newblock {\em {Nuclear Physics B}}, 492(3):682--716, 1997.

\bibitem{Berry2001}
M.~Berry and J.~P. Keating.
\newblock {Clusters of near-degenerate levels dominate negative moments of
  spectral determinants}.
\newblock {\em Journal of Physics A: Mathematical and General}, 35(1):L1, 2001.

\bibitem{Borodin2006}
A.~Borodin and E.~Strahov.
\newblock {Averages of characteristic polynomials in random matrix theory}.
\newblock {\em {Communications on Pure and Applied Mathematics}},
  59(2):161--253, 2006.

\bibitem{Breuer2012}
J.~Breuer and E.~Strahov.
\newblock {A universality theorem for ratios of random characteristic
  polynomials}.
\newblock {\em {Journal of Approximation Theory}}, 164(6):803--814, 2012.

\bibitem{Brezin2000}
E.~Br{\'e}zin and S.~Hikami.
\newblock {Characteristic polynomials of random matrices}.
\newblock {\em {Communications in Mathematical Physics}}, 214(1):111--135,
  2000.

\bibitem{Brezin2000edge}
E.~Br{\'e}zin and S.~Hikami.
\newblock Characteristic polynomials of random matrices at edge singularities.
\newblock {\em Physical Review E}, 62(3):3558, 2000.

\bibitem{Bump2006}
D.~Bump and A.~Gamburd.
\newblock {On the averages of characteristic polynomials from classical
  groups}.
\newblock {\em {Communications in Mathematical Physics}}, 265(1):227--274,
  2006.

\bibitem{Busbridge1939}
I.~W. Busbridge.
\newblock {The evaluation of certain integrals involving products of Hermite
  polynomials}.
\newblock {\em Journal of the London Mathematical Society}, 1(2):93--97, 1939.

\bibitem{Busbridge1948}
I.~W. Busbridge.
\newblock {Some integrals involving Hermite polynomials}.
\newblock {\em Journal of the London Mathematical Society}, 1(2):135--141,
  1948.

\bibitem{Conrey2003}
J.~B. Conrey, D.~W. Farmer, J.~P. Keating, M.~O. Rubinstein, and N.~C. Snaith.
\newblock {Autocorrelation of random matrix polynomials}.
\newblock {\em {Communications in Mathematical Physics}}, 237(3):365--395,
  2003.

\bibitem{Conrey2005integral}
J.~B. Conrey, D.~W. Farmer, J.~P. Keating, M.~O. Rubinstein, and N.~C. Snaith.
\newblock {Integral moments of L-functions}.
\newblock {\em Proceedings of the London Mathematical Society}, 91(1):33--104,
  2005.

\bibitem{Conrey2001}
J.~B. Conrey and S.~M. Gonek.
\newblock {High moments of the Riemann zeta-function}.
\newblock {\em Duke Mathematical Journal}, 107(3):577--604, 2001.

\bibitem{Damgaard1998}
P.~H. Damgaard and S.~M. Nishigaki.
\newblock {Universal spectral correlators and massive Dirac operators}.
\newblock {\em Nuclear Physics B}, 518(1-2):495--512, 1998.

\bibitem{Diaconis2004}
P.~Diaconis and A.~Gamburd.
\newblock Random matrices, magic squares and matching polynomials.
\newblock {\em the electronic journal of combinatorics}, pages R2--R2, 2004.

\bibitem{Forrester2004}
P.~Forrester and N.~Frankel.
\newblock {Applications and generalizations of Fisher--Hartwig asymptotics}.
\newblock {\em Journal of Mathematical Physics}, 45(5):2003--2028, 2004.

\bibitem{Forrester2004singularity}
P.~Forrester and J.~Keating.
\newblock {Singularity dominated strong fluctuations for some random matrix
  averages}.
\newblock {\em Communications in mathematical physics}, 250(1):119--131, 2004.

\bibitem{Forrester2010}
P.~J. Forrester.
\newblock {\em {Log-gases and random matrices (LMS-34)}}.
\newblock Princeton University Press, 2010.

\bibitem{Forrester2006}
P.~J. Forrester and A.~Gamburd.
\newblock Counting formulas associated with some random matrix averages.
\newblock {\em Journal of Combinatorial Theory, Series A}, 113(6):934--951,
  2006.

\bibitem{Fyodorov1995}
Y.~Fyodorov.
\newblock {Mesoscopic Quantum Physics ed E Akkermans et al}, 1995.

\bibitem{Fyodorov2002negative}
Y.~V. Fyodorov.
\newblock {Negative moments of characteristic polynomials of random matrices:
  Ingham--Siegel integral as an alternative to Hubbard--Stratonovich
  transformation}.
\newblock {\em Nuclear Physics B}, 621(3):643--674, 2002.

\bibitem{Fyodorov2018}
Y.~V. Fyodorov, J.~Grela, and E.~Strahov.
\newblock {On characteristic polynomials for a generalized chiral random matrix
  ensemble with a source}.
\newblock {\em {Journal of Physics A: Mathematical and Theoretical}},
  51(13):134003, 2018.

\bibitem{Fyodorov2012}
Y.~V. Fyodorov, G.~A. Hiary, and J.~P. Keating.
\newblock {Freezing transition, characteristic polynomials of random matrices,
  and the Riemann zeta function}.
\newblock {\em Physical review letters}, 108(17):170601, 2012.

\bibitem{Fyodorov2003negative}
Y.~V. Fyodorov and J.~P. Keating.
\newblock {Negative moments of characteristic polynomials of random GOE
  matrices and singularity-dominated strong fluctuations}.
\newblock {\em Journal of Physics A: Mathematical and General}, 36(14):4035,
  2003.

\bibitem{Fyodorov2014}
Y.~V. Fyodorov and J.~P. Keating.
\newblock Freezing transitions and extreme values: random matrix theory, and
  disordered landscapes.
\newblock {\em Philosophical Transactions of the Royal Society A: Mathematical,
  Physical and Engineering Sciences}, 372(2007):20120503, 2014.

\bibitem{Fyodorov2016moments}
Y.~V. Fyodorov and P.~Le~Doussal.
\newblock {Moments of the position of the maximum for GUE characteristic
  polynomials and for log-correlated Gaussian processes}.
\newblock {\em Journal of Statistical Physics}, 164(1):190--240, 2016.

\bibitem{Fyodorov2016}
Y.~V. Fyodorov and N.~J. Simm.
\newblock {On the distribution of the maximum value of the characteristic
  polynomial of GUE random matrices}.
\newblock {\em Nonlinearity}, 29(9):2837, 2016.

\bibitem{Fyodorov2002}
Y.~V. Fyodorov and E.~Strahov.
\newblock {Characteristic polynomials of random Hermitian matrices and
  Duistermaat--Heckman localisation on non-compact K{\"a}hler manifolds}.
\newblock {\em {Nuclear Physics B}}, 630(3):453--491, 2002.

\bibitem{Fyodorov2003}
Y.~V. Fyodorov and E.~Strahov.
\newblock {An exact formula for general spectral correlation function of random
  Hermitian matrices}.
\newblock {\em {Journal of Physics A: Mathematical and General}}, 36(12):3203,
  2003.

\bibitem{Garoni2005}
T.~Garoni.
\newblock {On the asymptotics of some large Hankel determinants generated by
  Fisher--Hartwig symbols defined on the real line}.
\newblock {\em Journal of mathematical physics}, 46(4):043516, 2005.

\bibitem{Hughes2000}
C.~P. Hughes, J.~P. Keating, and N.~O'Connell.
\newblock {Random matrix theory and the derivative of the Riemann zeta
  function}.
\newblock {\em Proceedings of the Royal Society of London. Series A:
  Mathematical, Physical and Engineering Sciences}, 456(2003):2611--2627, 2000.

\bibitem{Jonnadula2020}
B.~Jonnadula, J.~P. Keating, and F.~Mezzadri.
\newblock Symmetric function theory and unitary invariant ensembles.
\newblock {\em arXiv preprint arXiv:2003.02620}, 2020.

\bibitem{Keating2011}
J.~Keating and N.~Snaith.
\newblock Random matrix theory and number theory.
\newblock {\em The Handbook on Random Matrix Theory}, pages 491--509, 2011.

\bibitem{Keating2010}
J.~P. Keating, F.~Mezzadri, and B.~Singphu.
\newblock {Rate of convergence of linear functions on the unitary group}.
\newblock {\em {Journal of Physics A: Mathematical and Theoretical}},
  44(3):035204, 2010.

\bibitem{Keating2000}
J.~P. Keating and N.~C. Snaith.
\newblock {Random matrix theory and L-functions at $s= 1/2$}.
\newblock {\em {Communications in Mathematical Physics}}, 214(1):91--100, 2000.

\bibitem{Keating2000random}
J.~P. Keating and N.~C. Snaith.
\newblock {Random matrix theory and $\zeta$ (1/2+ it)}.
\newblock {\em {Communications in Mathematical Physics}}, 214(1):57--89, 2000.

\bibitem{Krasovsky2007}
I.~Krasovsky.
\newblock {Correlations of the characteristic polynomials in the Gaussian
  unitary ensemble or a singular Hankel determinant}.
\newblock {\em Duke Mathematical Journal}, 139(3):581--619, 2007.

\bibitem{Macdonald1998}
I.~G. Macdonald.
\newblock {\em Symmetric functions and Hall polynomials}.
\newblock Oxford university press, 1998.

\bibitem{Mehta2004}
M.~L. Mehta.
\newblock {\em Random matrices}.
\newblock Elsevier, 2004.

\bibitem{Mezzadri2005}
F.~Mezzadri, N.~C. Snaith, and N.~Hitchin.
\newblock {\em Recent perspectives in random matrix theory and number theory},
  volume 322.
\newblock Cambridge University Press, 2005.

\bibitem{Sergeev2014}
A.~N. Sergeev and A.~P. Veselov.
\newblock {Jacobi-Trudy formula for generalised Schur polynomials}.
\newblock {\em {Mosc. Math.~J.}}, 14(1):161--168, 2014.

\bibitem{Strahov2003moments}
E.~Strahov.
\newblock Moments of characteristic polynomials enumerate two-rowed
  lexicographic arrays.
\newblock {\em The electronic journal of combinatorics}, 10, 01 2002.

\bibitem{Strahov2003}
E.~Strahov and Y.~V. Fyodorov.
\newblock {Universal results for correlations of characteristic polynomials:
  Riemann-Hilbert approach}.
\newblock {\em {Communications in Mathematical Physics}}, 241(2-3):343--382,
  2003.

\bibitem{Szabo2001}
R.~J. Szabo.
\newblock {Microscopic spectrum of the QCD Dirac operator in three dimensions}.
\newblock {\em Nuclear Physics B}, 598(1-2):309--347, 2001.

\end{thebibliography}
\end{document}